\documentclass[runningheads]{llncs}

\usepackage[english]{babel}
\usepackage{tabularx,lipsum}
\usepackage{pdfsync}
\usepackage{amssymb}
\usepackage{xspace}
\usepackage{amsmath}
\usepackage{latexsym}
\usepackage{epsfig}
\usepackage{subfigure}
\usepackage{times}
\usepackage{enumerate}
\usepackage{tabularx}
\usepackage{caption}
\usepackage[plainpages=false]{hyperref}
\usepackage{multirow}
\usepackage[usenames,dvipsnames]{color}
\usepackage{hyphenat}
\usepackage{amsfonts}
\usepackage{dsfont}
\usepackage{paralist}
\usepackage{mfirstuc}
\usepackage[color=yellow,textsize=scriptsize]{todonotes}
\usepackage{marginnote}
\usepackage{cooltooltips}

\usepackage{tipa}
\usepackage[full]{textcomp}
\usepackage{ upgreek }
\newcommand{\remove}[1]{}

\definecolor{blue}{rgb}{0.274,0.392,0.666}
\definecolor{red}{rgb}{0.627,0.117,0.156}
\definecolor{green}{rgb}{0,0.588,0.509}

\makeatletter
\@ifclassloaded{elsarticle}{

\newtheorem{theorem}{Theorem}

\newtheorem{corollary}{Corollary}
\newenvironment{proof}{{\em Proof.}}
}{}
\makeatother

\newcommand{\red}[1]{{\color{red}{#1\xspace}}}
\newcommand{\blue}[1]{{\color{blue}{#1\xspace}}}
\newcommand{\NP}{$\mathcal{NP}$\xspace}

\newcommand{\NPC}{\mbox{\NP-complete}\xspace}

\newcommand{\Er}{\textcolor{red}{$E_1$}\xspace}
\newcommand{\Eb}{\textcolor{blue}{$E_2$}\xspace}

\newcommand{\Gint}{$G_{\cap}$\xspace}

\newcommand{\Gr}{\textcolor{red}{$G_1$}\xspace}
\newcommand{\GammaR}{\textcolor{red}{$\Gamma_1$}\xspace}

\newcommand{\Gb}{\textcolor{blue}{$G_2$}\xspace}
\newcommand{\GammaB}{\textcolor{blue}{$\Gamma_2$}\xspace}

\newcommand{\sefeinstance}{$\langle\textcolor{red}{G_1},\textcolor{blue}
{G_2}\rangle$\xspace}

\newcommand{\cpp}{{\scshape C-Planarity}\xspace}
\newcommand{\sefep}{{\scshape SEFE-$2$}\xspace}
\newcommand{\cksefep}{{\scshape C-SEFE-$k$}\xspace}
\newcommand{\csefep}{{\scshape C-SEFE-$2$}\xspace}

\newcommand{\pTHREEpbepshort}{{\scshape PBE-$3$}\xspace}

\newcommand{\ptckpbepshort}{{\scshape PTBE-$k$}\xspace}
\newcommand{\ptcTWOpbepshort}{{\scshape PTBE-$2$}\xspace}

\newcommand{\ptcTWOpbeinstance}{$\langle T, \textcolor{red}{E_1},\textcolor{blue}{E_2} \rangle$\xspace}

\newcommand{\sefesolution}{$\langle\red{\Gamma_1},\blue{ \Gamma_2 } \rangle$\xspace}

\newif\ifllncsvar

\makeatletter \@ifclassloaded{llncs}{ \llncsvartrue \author{Patrizio
    Angelini and Giordano {Da Lozzo}}

  \title{Deepening the Relationship between\\ SEFE and C-Planarity}
  \institute{
    Department of Engineering, Roma Tre University, Italy\\
    \email{\{angelini,dalozzo\}@dia.uniroma3.it} }
  \bibliographystyle{splncs03}}{
  \usepackage{amssymb} \usepackage[figuresright]{rotating}

}
\makeatother

\let\nu\undefined \DeclareMathOperator{\nu}{\theta}

\authorrunning{Angelini and {Da Lozzo}}
\titlerunning{Deepening the Relationship between SEFE and C-Planarity}

\begin{document}
\ifllncsvar
\maketitle
\else
\begin{frontmatter}
\fi

\begin{abstract} 
  In this paper we deepen the understanding of the connection
  between two long-standing Graph Drawing open problems, that is,
  Simultaneous Embedding with Fixed Edges ({\sc SEFE}) and Clustered
  Planarity (\cpp).  In his GD'12 paper Marcus Schaefer presented a
  reduction from \cpp to SEFE of two planar graphs (\sefep).  We prove
  that a reduction exists also in the opposite direction, if we
  consider instances of \sefep in which the intersection graph is
  connected.  We pose as an open question whether the two problems are
  polynomial-time equivalent.
\end{abstract}

\ifllncsvar \else

\author{Patrizio
  Angelini}
\author{Giordano {Da
    Lozzo}}


\address{{\small Department of Engineering, Roma Tre University, Italy}\\
  {\{angelini,dalozzo\}@dia.uniroma3.it}}

\begin{keyword} {\sc SEFE} \sep \cpp \sep Graph Drawing \sep
  computational complexity \sep polynomial-time reducibility
\end{keyword}

\end{frontmatter}
\fi

\section{Introduction}\label{se:introduction}

In recent years the problem of displaying together multiple
relationships among the same set of entities has turned into a central
subject of research in Graph Drawing and Visualization. In this
context, the two major paradigms that held the stage
are the simultaneous embedding of graphs, in which the relationships
are described by means of different sets of edges among the same set
of vertices, and the visualization of clustered graphs, in which the
relationships are described by means of a set of edges and of a
cluster hierarchy grouping together vertices with semantic affinities.

We study the connection between two problems adhering to such
paradigms, {\em Simultaneous Embedding with Fixed Edges} ({\sc SEFE}) and
{\em Clustered Planarity} (\cpp). Given $k$ graphs
$G_1(V,E_1),\dots,G_k(V,E_k)$ the {\sc SEFE-$k$} problem asks whether
there exist $k$ planar drawings $\Gamma_i$ of $G_i$, with
$i=1,\dots,k$, such that: (i) any vertex $v\in V$ is mapped to the
same point in any $\Gamma_i$; (ii) any edge $e \in E_i \cap E_j$ is
mapped to the same curve in $\Gamma_i$ and $\Gamma_j$
(see~\cite{bkr-sepg-12} for a comprehensive survey on this topic).
Given a graph $G(V,E)$ and a cluster hierarchy over $V$,
the \cpp problem asks whether a planar drawing of $G$ exists such that
each cluster can be drawn as a simple region enclosing all and only
the vertices belonging to it without introducing unnecessary
intersections involving clusters and edges
(see~\cite{FengCE95,df-ectefcgsf-09}).

Due to their practical relevance and their theoretical appealing,
these problems have attracted a great deal of effort in the research
community. However, despite several restricted cases have been
successfully settled, the question regarding the computational
complexity of the original problems keeps being as elusive as their
charm.

In a recent work~\cite{s-ttphtpv-13}, Marcus Schaefer leveraged the
expressive power of {\sc SEFE-$k$} to generalize, in terms of
polynomial-time reducibility, several graph drawing problems,
including \cpp.  On the other hand, also \cpp has shown a significant
expressive power as it generalizes relevant problems, like {\sc Strip
  Planarity}~\cite{addf-spt-13}; most notably, two special cases of
\cpp and \sefep have been proved to be polynomial-time
equivalent~\cite{hn-tpbecgp-09,adfpr-tsetgibgt-11}, that is, \cpp with
two clusters and \sefep where the common graph is a star. Motivated by
such results, we pose the question whether this equivalence extends to
the general case. 

In this paper we take a first step in this
direction, by proving that \csefep, that is the restriction of \sefep
to instances with connected common graph, reduces to \cpp.

The paper is structured as follows. In Section~\ref{se:preliminaries}
we give basic definitions. In Section~\ref{se:reduction} we show a
polynomial-time reduction from \csefep to \cpp. Finally, in
Section~\ref{se:conclusions} we give conclusive remarks and present
some open problems.

\section{Preliminaries}\label{se:preliminaries}
\ifllncsvar
A graph $G=(V,E)$ is a pair, where $V$ is the set of vertices and
$E\subseteq V^2$ is the set of edges. A graph without self-loops and
multi-edges is called {\em simple}. In the following, we will only
consider simple graphs. The {\em degree} of a vertex is the number of
edges incident to it.

A graph is \emph{connected} if every two vertices are connected by a
path. A $tree$ $T$ is a minimally connected graph. Namely, for each
two vertices there exists exactly one path connecting them. The
degree-$1$ vertices of $T$ are {\em leaves}, while the other vertices
are {\em internal} vertices. We denote the set of leaves by
$\mathcal{L}(T)$.

A \emph{drawing} of a graph is a mapping of each vertex to a point of
the plane and of each edge to a simple curve connecting its endpoints.
A drawing is \emph{planar} if the curves representing its edges do not
cross except, possibly, at common endpoints. A graph is \emph{planar}
if it admits a planar drawing.  A planar drawing partitions the plane
into topologically connected regions called {\em faces}.  Two planar
drawings are said to be equivalent if they determine the same circular
order of edges around each vertex.  An equivalence class of planar
drawings is called an {\em embedding}.  \else We assume familiarity
with basic definitions of planar graphs and embeddings
(see~\cite{dett-gd-99,tamax-hgdv-07}).  \fi

Given $k$ planar graphs $G_1(V,E_1),\dots,G_k(V,E_k)$ such that $E_i
\cap E_j = \emptyset$, with $1\leq i<j\leq k$, a {\em $k$-page
  book-embedding} of graphs $G_i$ consists of a linear ordering
$\mathcal{O}$ of $V$ such that for every set $E_i$ there exist no two
edges $e_1,e_2 \in E_i$ whose endvertices alternate in $\mathcal{O}$.
The {\sc Partitioned T-Coherent $k$-Page Book Embedding} problem
(\ptckpbepshort) takes as input a rooted tree $T$ with leaves
$\mathcal{L}(T)$ and $k$ sets $E_1,\dots,E_k$ of edges among leaves such that
$E_i \cap E_j = \emptyset$, with $1\leq i<j\leq k$, and asks whether a
$k$-page book-embedding $\mathcal{O}$ of graphs $G_i =
(\mathcal{L}(T), E_i)$ exists such that $\mathcal{O}$ is represented
by $T$.  It is easy to verify that instance $\langle T, E_1,\dots,E_k
\rangle$ admits a partitioned T-coherent $k$-page book-embedding if
and only if graphs $G_i=(V(T),E(T) \cup E_i)$ admit a \cksefep.

Since \csefep and \ptcTWOpbepshort have been proved to be
polynomial-time equivalent~\cite{adfpr-tsetgibgt-11}, in order to simplify the
description, in the following we will denote an instance
\ptcTWOpbeinstance of \ptcTWOpbepshort by the corresponding instance
\sefeinstance of \csefep, where $\red{G_1}=(V(T),E(T) \cup \red{E_1})$
and $\blue{G_2}=(V(T),E(T) \cup \blue{E_2})$, and vice versa.

The \cpp problem, introduced by Feng {\em et al.}~\cite{fce-hdpcg-95},
takes as input a {\em clustered graph} $C(G,\mathcal{T})$, that is a
planar {\em underlying graph} $G$ together with a {\em cluster
  hierarchy} $\mathcal{T}$, that is a rooted tree whose leaves are the
vertices of $G$. Each internal node $\mu$ of $\mathcal{T}$ is called
{\em cluster} and is associated with the leaves of the subtree
$\mathcal{T}(\mu)$ of $\mathcal{T}$ rooted at $\mu$.  The problem asks
whether a {\em $c$-planar drawing} of $C(G,\mathcal{T})$ exists, that
is a planar drawing of $G$ together with a drawing of each cluster
$\mu$ as a simple region $R(\mu)$ such that:
\begin{inparaenum}
\item $R(\mu)$ encloses all and only the leaves of $\mathcal{T}(\mu)$
  and the regions representing the internal nodes of
  $\mathcal{T}(\mu)$;
\item $R(\mu) \cap R(\nu) \neq \emptyset$ if and only if $\nu$ is an
  internal node of $\mathcal{T}(\mu)$; and
\item each edge $(u,v)$ of $G$ intersects $R(\mu)$ at most once.
\end{inparaenum}
A clustered graph $C(G,\mathcal{T})$ is {\em flat} if $\mathcal{T}$ is
a tree of height $2$ (that is, removing all the leaves yields a star
graph) and {\em non-flat} otherwise.

\section{Reduction}\label{se:reduction}

In this section we prove the main result of the paper. To ease the
description, we first prove in Theorem~\ref{th:non-flat} that \csefep
reduces to \cpp, where the constructed instance of \cpp is
non-flat. We give an high level view of the reduction. Due to the
equivalence between \csefep and
\ptcTWOpbepshort~\cite{adfpr-tsetgibgt-11}, the reduction is performed
on instances of \ptcTWOpbepshort. The proof exploits two clusters to
enforce the placement of the edges of the reduced instance on the
pages they are assigned to (similarly to the technique used
in~\cite{hn-tpbecgp-09} to reduce \ptcTWOpbepshort in which $T$ is a
star to \cpp) and a suitable cluster hierarchy to represent the
constraints on the book-embedding imposed by $T$. Then,
Theorem~\ref{th:flat} states that the reduction of
Theorem~\ref{th:non-flat} can be extended to obtain flat instances
whose underlying graph is a set of paths.

\begin{theorem}\label{th:non-flat}
  \csefep $\propto$ \cpp.
\end{theorem}

\begin{proof}
  Let \ptcTWOpbeinstance be an instance of \ptcTWOpbepshort
  (corresponding to instance \sefeinstance of \csefep) and let $r$ be
  the root of $T$. We describe how to construct an equivalent instance
  $C(G,\mathcal{T})$ of \cpp starting from \ptcTWOpbeinstance. Refer
  to Fig~\ref{fig:cplanarity}.

\begin{figure}[tb]
  \centering
  \subfigure[\ptcTWOpbeinstance]{\includegraphics[width=.39\textwidth]{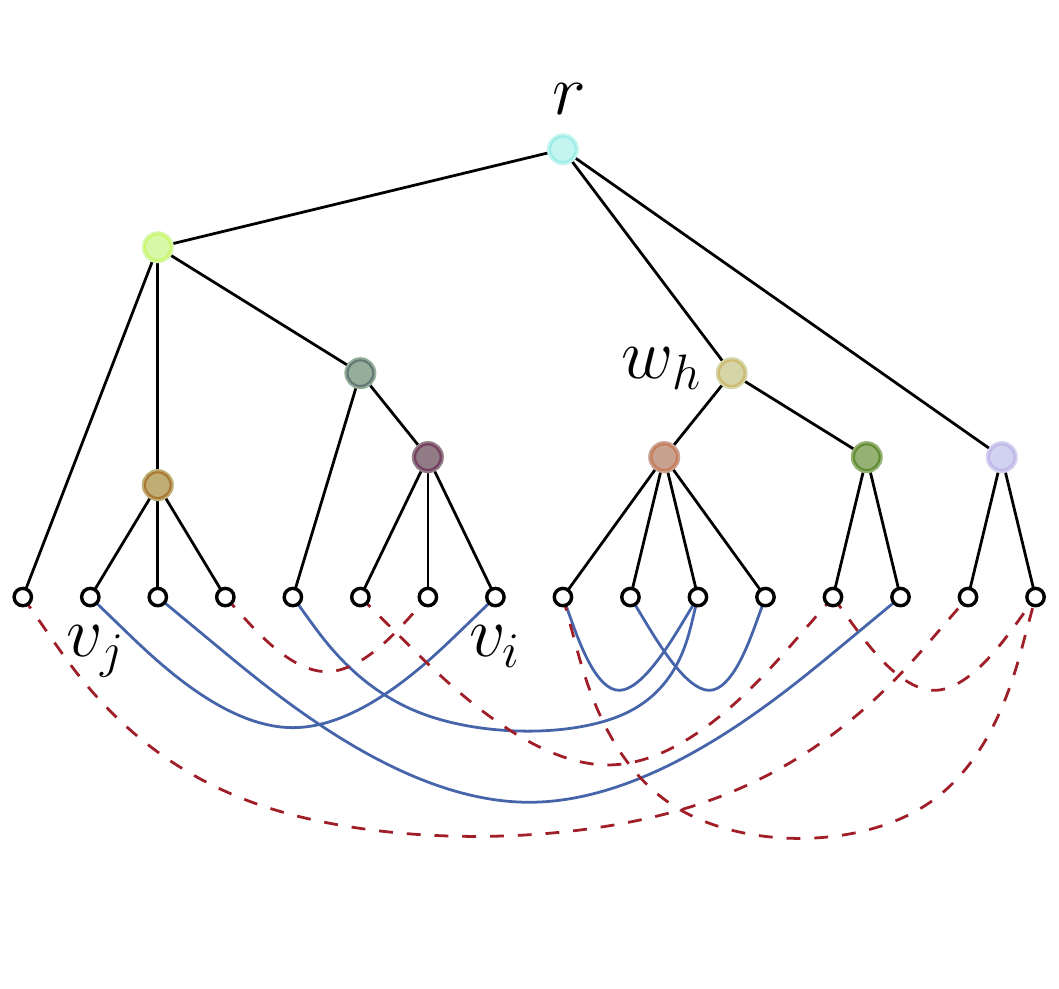}\label{fig:sefeinstance}}
  \subfigure[$C(G,\mathcal{T})$]{\includegraphics[width=.6\textwidth]{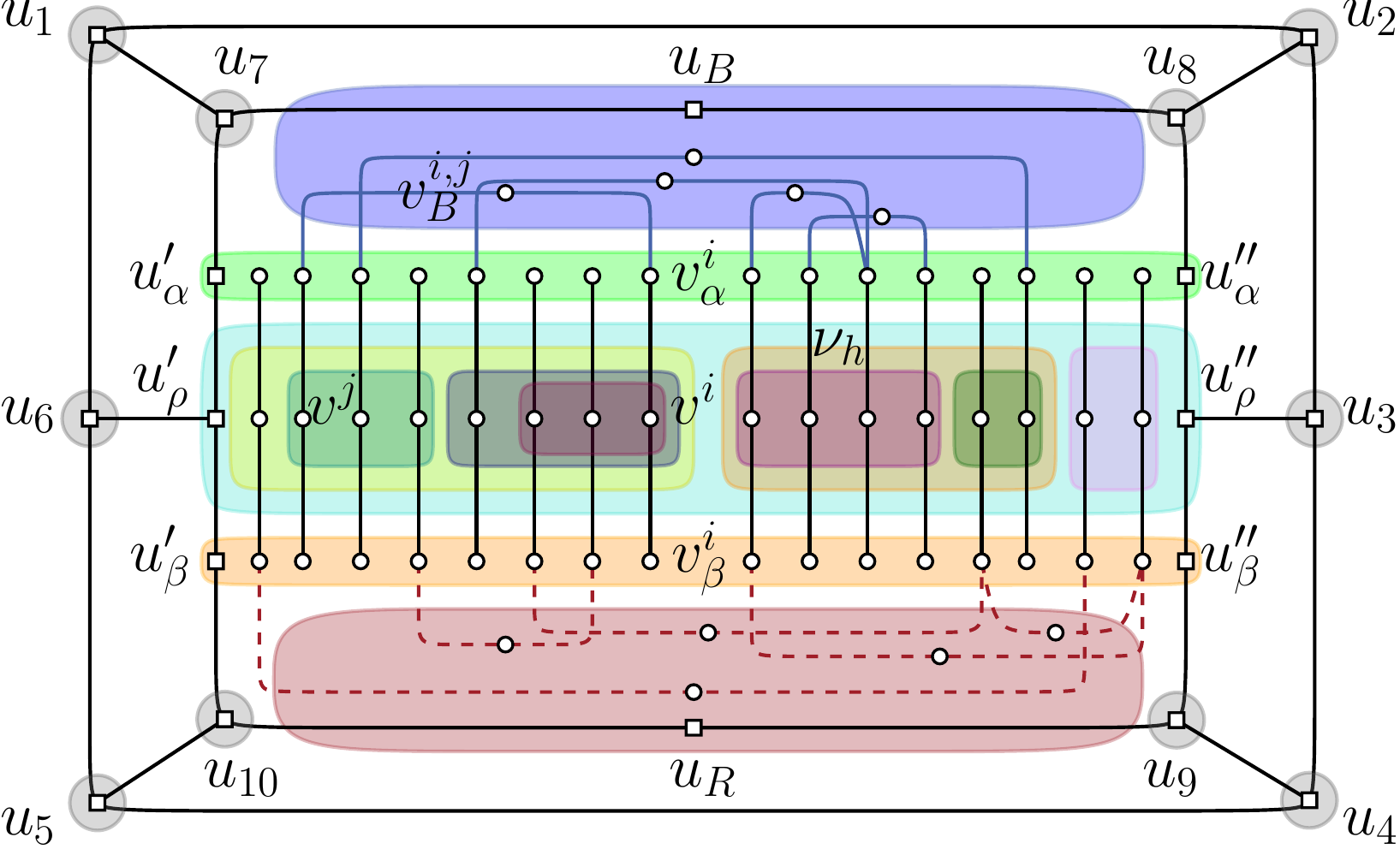}\label{fig:cppinstance}}
  \caption{Illustration of the reduction from \csefep to
    \cpp. Correspondence between internal vertices of $T$ and clusters
    of $\mathcal{T}$ is encoded with colors. The root cluster
    $\lambda$ is not represented in (b).}\label{fig:cplanarity}
\end{figure}

Initialize $G$ to a graph $H$ composed of two cycles $C_1=\langle
u_1,u_2,u_3,u_4,u_5,u_6 \rangle$ and $C_2=\langle
u_7,u_B,u_8,u''_\alpha,u''_\rho,u''_\beta,u_9,u_R,u_{10},u'_\beta,u'_\rho,u'_\alpha
\rangle$, and of edges $(u_1,u_7)$, $(u_2,u_8)$, $(u_3,u''_\rho)$,
$(u_4,u_9)$, $(u_5,u_{10})$, and $(u_6,u'_\rho)$. Observe that $H$ is
a subdivision of a $3$-connected planar graph.

Initialize $\mathcal{T}$ to a tree only composed of a root $\lambda$.
For $m=1,\dots,10$, add a cluster $\mu_m$ to $\mathcal{T}$ as a child
of $\lambda$, containing only vertex $u_m$. Also, add clusters $\mu_B$
and $\mu_R$ to $\mathcal{T}$ as children of $\lambda$, containing
vertices $u_B$ and $u_R$, respectively. Finally, for $\sigma \in
\{\alpha,\rho,\beta\}$, add a cluster $\mu_\sigma$ to $\mathcal{T}$ as
a child of $\lambda$, containing vertices $u'_\sigma$ and
$u''_\sigma$.

Then, consider each internal vertex $w_h$ of $T$ according to a
top-down traversal of $T$ and add to $\mathcal{T}$ a cluster $\nu_h$
either as a child of cluster $\nu_k$, if $w_k \neq r$ is the parent of
$w_h$ in $T$, or as a child of cluster $\mu_\rho$, if $r$ is the
parent of $w_h$ in $T$.  Also, for each leaf vertex $v_i$ of $T$, add
to $G$ a path $(v_\alpha^i, v^i, v_\beta^i)$, that we call {\em
  leaf-path}. Add vertices $v_\alpha^i$ and $v_\beta^i$ to clusters
$\mu_\alpha$ and $\mu_\beta$, respectively; add $v^i$ to cluster
$\nu_h$, if $w_h$ is the parent of $v_i$ in $T$, or to cluster
$\mu_\rho$, if $r$ is the parent of $v_i$ in $T$.

Finally, for each edge $(v_i,v_j)$ in \Er or in \Eb, add to $G$ path
$(v_\beta^i,v_R^{i,j},v_\beta^j)$ or path
$(v_\alpha^i,v_B^{i,j},v_\alpha^j)$, respectively, that we call {\em
  edge-paths}.  Add each vertex $v_R^{i,j}$ to $\mu_R$ and each vertex
$v_B^{i,j}$ to $\mu_B$.

Suppose that \ptcTWOpbeinstance admits a {\sc SEFE} \sefesolution. We
show how to construct a c-planar drawing $\Gamma$ of
$C(G,\mathcal{T})$. We will construct the drawing of $G$ contained in
$\Gamma$ as a straight-line drawing; hence, we only describe how to
place the vertices of $G$.  \ifllncsvar Refer to
Fig.~\ref{fig:sefe-cp-drawing}.

\begin{figure}[htb]
  \centering
  \includegraphics[width=.45\textwidth]{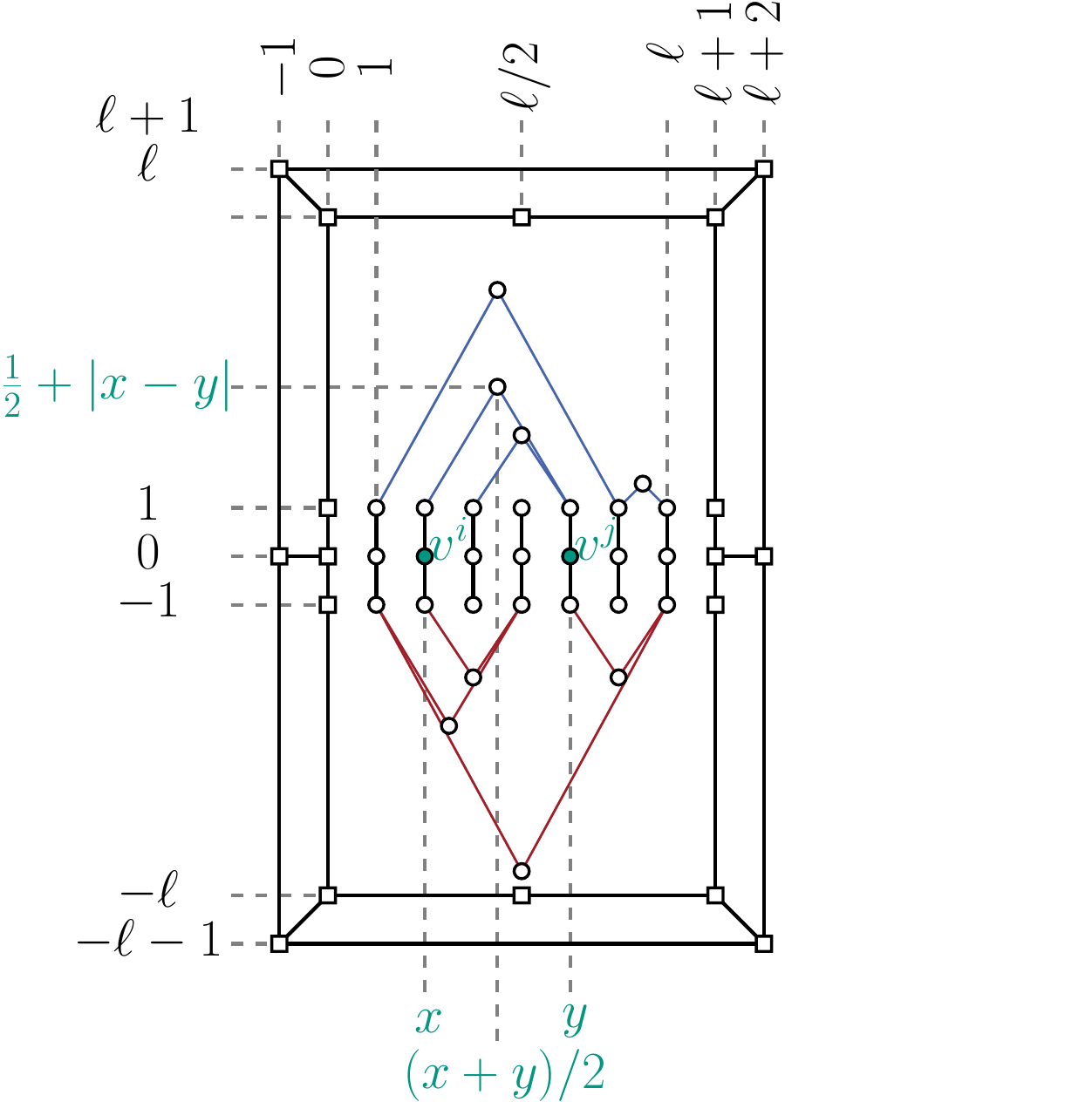}
  \caption{Construction of a c-planar drawing of $C(G,\mathcal{T})$
    starting from \sefesolution, where $x=\phi(v_i)$ and
    $y=\phi(v_j)$.}\label{fig:sefe-cp-drawing}
\end{figure}
\fi

Let $\ell=|\mathcal{L}(T)|$. We first consider cycle $C_1$. Place
vertex $u_1$ at point $(-1,\ell+1)$, $u_2$ at $(\ell+2,\ell+1)$, $u_3$
at $(\ell+2,0)$, $u_4$ at $(\ell+2,-\ell-1)$, $u_5$ at $(-1,-\ell-1)$,
and $u_6$ at $(-1,0)$. Then, we consider cycle $C_2$. Place vertex
$u_7$ at point $(0,\ell)$, $u_B$ at $(\frac{\ell}{2},\ell)$, $u_8$ at
$(\ell+1,\ell)$, $u''_\alpha$ at $(\ell+1,1)$, $u''_\rho$ at
$(\ell+1,0)$, $u''_\beta$ at $(\ell+1,-1)$, $u_9$ at $(\ell+1,-\ell)$,
$u_R$ at $(\frac{\ell}{2},-\ell)$, $u_{10}$ at $(0,-\ell)$, $u'_\beta$
at $(0,-1)$, $u'_\rho$ at $(0,0)$, $u'_\alpha$ at $(0,1)$.

Consider the circular order of the leaves of $T$ determined by
\sefesolution and consider two adjacent leaves $v'$ and $v''$ such
that the lowest common ancestor of $v'$ and $v''$ in $T$ is the root
$r$ (note that, if $r$ has degree greater than $1$, there always exist
two such vertices; otherwise, we can obtain an equivalent instance of
SEFE by removing $r$ from $T$). Consider the linear order
$\mathcal{O}$ of the leaves of $T$ such that $v'$ and $v''$ are the
first and the last element of $\mathcal{O}$.  Let
$\phi:\mathcal{L}(T)\rightarrow{1,\dots,\ell}$ be a function such that
$\phi(v_i)=k$ if $v_i$ is the $k$-th element in $\mathcal{O}$. For
each leaf vertex $v_i$, we draw leaf-path $(v_\alpha^i, v^i,
v_\beta^i)$ by placing vertex $v^i$ at point $(x,0)$, $v_\alpha^i$ at
$(x,1)$, and $v_\beta^i$ at $(x,-1)$, where $x=\phi(v_i)$.  Then, for
each edge $(v_i,v_j) \in \blue{E_2}$, we draw edge-path
$(v^i,v^{i,j}_B,v^j)$ by placing vertex $v_B^{i,j}$ at point
$(\frac{x+y}{2},\frac{1}{2}+|x-y|)$, where $x=\phi(v_i)$ and
$y=\phi(v_j)$. Symmetrically, for each edge $(v_i,v_j) \in \red{E_1}$,
we draw edge-path $(v^i,v^{i,j}_R,v^j)$ by placing vertex $v_R^{i,j}$
at point $(\frac{x+y}{2},-\frac{1}{2}-|x-y|)$, where $x=\phi(v_i)$ and
$y=\phi(v_j)$.

Finally, we draw the region representing each cluster. Consider each
cluster $\mu \in \mathcal{T}$ according to a bottom-up traversal and
draw $\mu$ as an axis-parallel rectangular region enclosing all and
only the vertices and clusters in the subtree of $\mathcal{T}$ rooted
at $\mu$. Observe that, this is always possible. Namely, for clusters
$\mu_m$, with $m=1,\dots,10$, and clusters $\mu_B$, $\mu_R$,
$\mu_\alpha$, and $\mu_\beta$ this directly follows from the
construction. Also, for each cluster $\nu_h$ corresponding to an
internal vertex $w_h$ of $T$, this descends from the fact that the
ordering of the leaves of $T$ is determined by a SEFE
\sefesolution. Indeed, since the drawing of $T$ is planar in
\sefesolution, for any two vertices $v^i$ and $v^j$ of $G$ belonging
to the same cluster, there exists no vertex $v^k$, with
$\phi(v_i)<\phi(v_k)<\phi(v_j)$, belonging to a different cluster.
Since all leaf-paths are drawn as vertical segments, this implies that
no edge-region crossing occurs between a cluster $\nu_h$ and a
leaf-path. Once all clusters $\nu_h$ have been drawn, cluster
$\mu_\rho$ can be drawn to enclose all and only such clusters.

Further, observe that there exist no two edge-paths
$(v^i_\alpha,v^{i,j}_B,v^j_\alpha)$ and
$(v^h_\alpha,v^{h,k}_B,v^k_\alpha)$, corresponding to edges
$(v_i,v_j)$ and $(v_p,v_q)$ of \Eb, such that pairs $\langle
v_i,v_j\rangle$ and $\langle v_p, v_q\rangle$ alternate in
$\mathcal{O}$. Hence, any two edge-paths are either disjoint or
nested. In both cases, by construction, they do not cross (see
Fig.~\ref{fig:sefe-cp-drawing} for an illustration of the two cases).
Similarly, it can be proved that edge-paths corresponding to edges of
\Er do not cross.  This concludes the proof that $\Gamma$ is a
c-planar drawing of $C(G,\mathcal{T})$.

Suppose that $C(G,\mathcal{T})$ admits a c-planar drawing $\Gamma$. We
show how to construct a SEFE \sefesolution of \ptcTWOpbeinstance.
First, observe that all leaf-paths entirely lie inside the face $f$ of
$H$ delimited by cycle $C_2$, as $f$ is the only face of $H$ shared by
$u'_\alpha$, $u'_\rho$, $u'_\beta$, $u''_\alpha$, $u''_\rho$, and
$u''_\beta$. Since all vertices $v^{i,j}_B$ and $v^{i,j}_R$ are
adjacent to vertices of leaf-paths, they also lie inside $f$. Further,
since for $\sigma \in \{\alpha,\rho,\beta\}$ cluster $\mu_\sigma$ is
represented by a connected region enclosing vertices $u'_\sigma$ and
$u''_\sigma$ and not involved in any edge-region and region-region
crossing, all the edges connecting vertices of $\mu_\sigma$ to
vertices of the same cluster are consecutive in the order of the edges
crossing the boundary of $\mu_\sigma$.  This implies that the order in
which leaf-paths cross the boundary of $\mu_\alpha$ is the reverse of
the order in which they cross the boundary of $\mu_\beta$, since no
two leaf-paths cross each other in $\Gamma$.  To obtain \sefesolution,
we order the leaves $v_i$ of $T$ according to the order in which
leaf-paths cross the boundary of $\mu_\alpha$. Let $\mathcal{O}$ be
such an order.

First, we show that $\mathcal{O}$ can be represented by $T$, which
implies that a planar drawing $\Gamma_T$ of $T$ exists respecting
$\mathcal{O}$. In fact, by construction, for each internal vertex
$w_h$ of $T$, the leaves of the subtree $T(w_h)$ of $T$ rooted at
$w_h$ belong to the same cluster $\nu_h$. Also, since $\Gamma$ is
c-planar, all the leaf-paths $(v^i_\alpha,v^i,v^i_\beta)$ such that
$v_i$ is a leaf of $T(w_h)$ are consecutive in the order in which
leaf-paths cross the boundary of $\mu_\alpha$ and hence the
corresponding leaves $v_i$ are consecutive in $\mathcal{O}$.  Second,
we show how to construct two planar drawings \GammaR and \GammaB of
\Gr and \Gb, respectively, such that the drawing of $T$ contained in
\GammaR and in \GammaB coincides with $\Gamma_T$.  We describe the
algorithm to construct \GammaB, the algorithm for \GammaR being
analogous. Consider two edges $(v_i,v_j)$ and $(v_p,v_q)$ of \Eb.
Since the drawing of $G$ in $\Gamma$ is planar, the corresponding
edge-paths $(v^i_\alpha,v^{i,j}_B,v^j_\alpha)$ and
$(v^p_\alpha,v^{p,q}_B,v^q_\alpha)$ do not intersect in $\Gamma$.
Also, since the edges belonging to edge-paths are consecutive in the
order in which edges incident to vertices of $\mu_\alpha$ cross the
boundary of $\mu_\alpha$, the pair of leaves $\langle v_i,v_j\rangle$
and $\langle v_p, v_q\rangle$ of $T$ corresponding to vertices
$v^i_\alpha$, $v^j_\alpha$, $v^p_\alpha$, and $v^q_\alpha$ do not
alternate in $\mathcal{O}$.  Hence, \GammaB can be obtained from
$\Gamma_T$ by drawing the edges of \Eb as curves intersecting neither
edges of $T$ nor other edges in \Eb.  Since the drawing of $G_\cap=T$
is the same in \GammaR and in \GammaB, \sefesolution is a SEFE of
\ptcTWOpbeinstance.  This concludes the proof of the theorem. \qed
\end{proof}

\ifllncsvar
In the following we prove that the reduction of Theorem~\ref{th:non-flat} can be modified in such a way that the resulting instance of \cpp is flat and the underlying graph consists of a set of paths.
\fi

\begin{theorem}\label{th:flat}
  \csefep $\propto$ \cpp with flat cluster hierarchy and underlying
  graph that is a set of paths.
\end{theorem}

\newcommand{\proofthflat}{ Let \ptcTWOpbeinstance be an instance of
  \csefep. We describe how to construct an equivalent instance
  $C(G,\mathcal{T})$ of \cpp with flat cluster hierarchy and
  underlying graph that is a set of paths starting from
  \ptcTWOpbeinstance.

  First, we construct an instance $C^*(G^*,\mathcal{T}^*)$ of \cpp
  with non-flat cluster hierarchy by applying the reduction shown in
  Theorem~\ref{th:non-flat}.  We describe how to transform $C^*$ into
  an equivalent instance $C(G,\mathcal{T})$ of \cpp with the required
  properties.

  For vertices $u'_\rho$, $u''_\rho$, and for all vertices
  $v^i_\alpha$ and $v^i_\beta$ having degree at least $2$, consider
  the parent cluster $\nu$ of any such vertex $v$ in
  $\mathcal{T}$. Add a cluster $\mu_v$ to $\mathcal{T}$ as a child of
  $\nu$ and containing only vertex $v$. The obtained instance
  $C'(G'=G,\mathcal{T}')$ is obviously equivalent to $C^*$.

  Let $\Delta$ be the set of all clusters $\tau \in \mathcal{T}'$ such
  that $\mathcal{T}'(\tau)$ has only one leaf $t$. Note that, $\Delta$
  consists of all clusters $\mu_m$, with $m=1,\dots,10$, and all
  clusters added at the previous step to obtain $C'$. For each cluster
  $\tau \in \Delta$, we perform the following procedure. For each edge
  $(t,z)$ of $G'$ such that $t \in \tau$, add a vertex $t_z$ to $\tau$
  and add edge $(t_z,z)$ to $G'$. Finally, remove vertex $t$ and its
  incident edges from $C'$. This can be seen as replacing $t$ with
  $\deg(t)$ copies of it. For simplicity, in the following we keep the
  same notation $(v^i_\alpha,v^i,v^i_\beta)$ for leaf-paths, and
  $(v^i_\alpha,v^{i,j}_B,v^j_\alpha)$ and $(v^i_\beta,v^{i,j}_R,v^j_\beta)$ for
  edge-paths, where their endvertices have been
  naturally replaced by the appropriate copy.  See
  Fig.~\ref{fig:A-paths} for an illustration of this step.

\begin{figure}[tb]
  \centering
  \subfigure[$C''(G'',\mathcal{T}'')$]{\includegraphics[width=.49\textwidth]{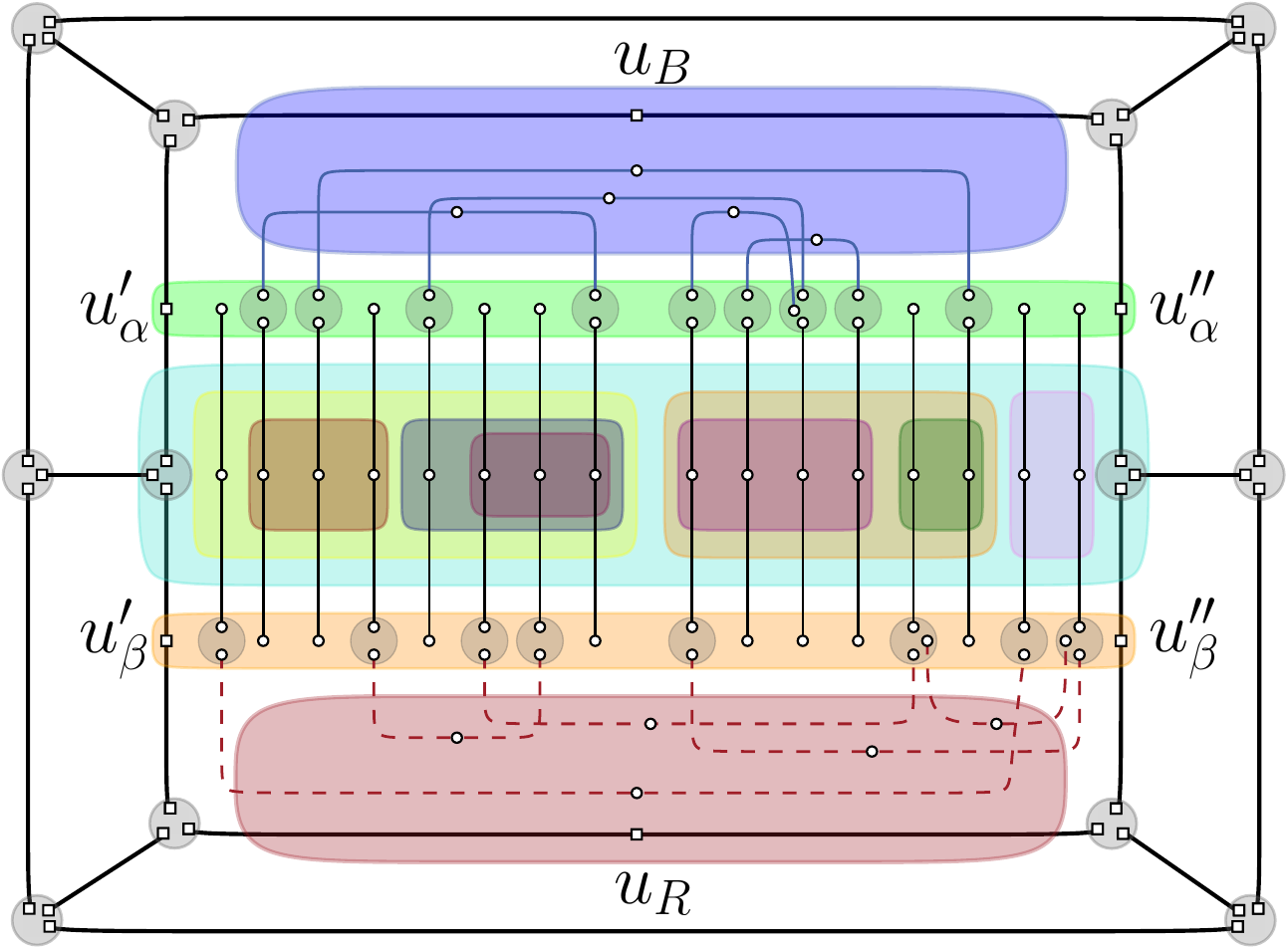}\label{fig:A-paths}}
  \subfigure[$C(G,\mathcal{T})$]{\includegraphics[width=.49\textwidth]{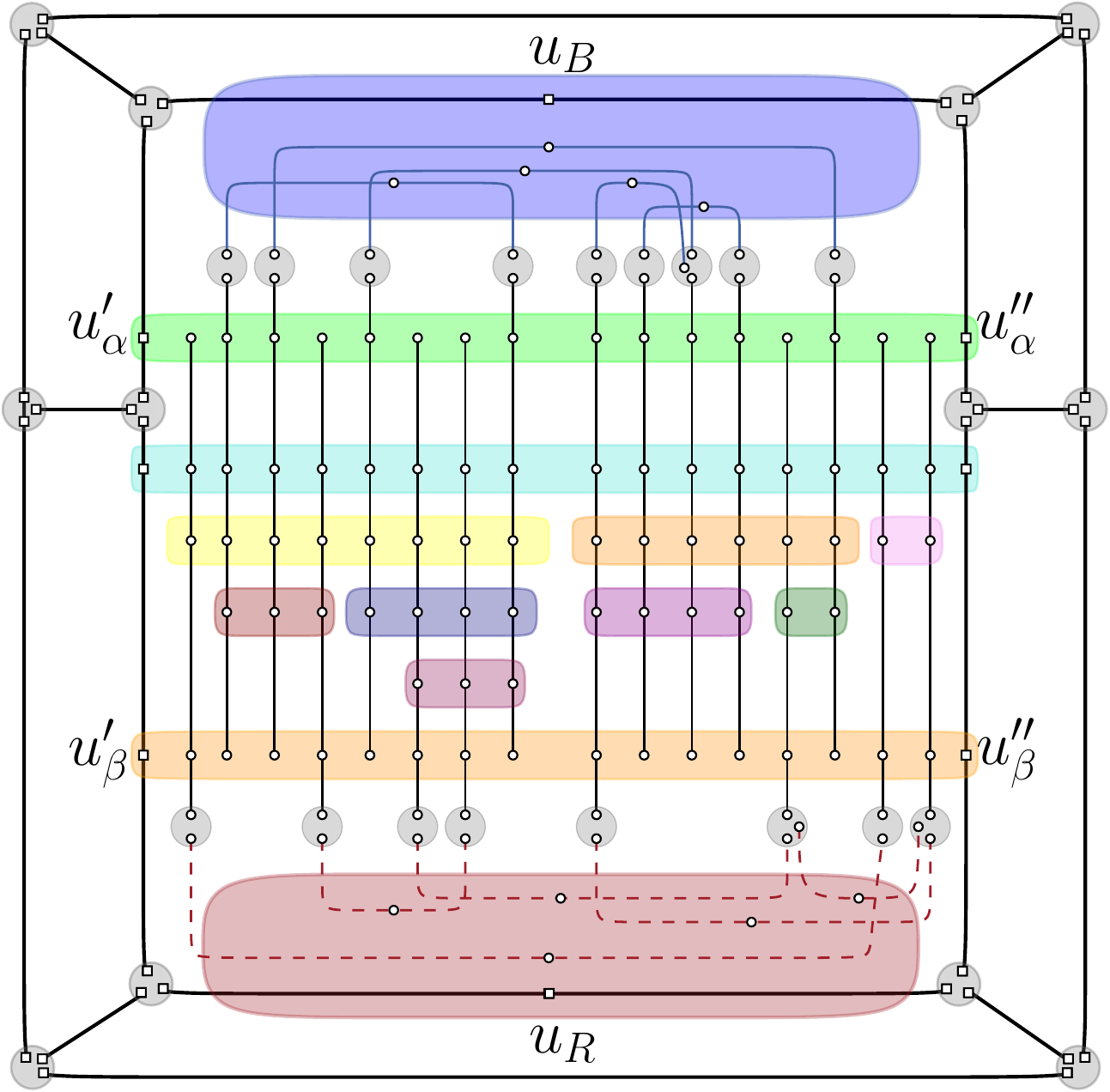}\label{fig:B-flat}}
  \caption{Construction of an equivalent instance $C(G,\mathcal{T})$
    with the desired properties. (a) Obtaining an instance whose
    underlying graph is a set of paths. (b) Obtaining a flat
    instance.}
\end{figure}

Observe that, the constructed instance $C''(G'',\mathcal{T}'')$ is
such that $G''$ consists of a set of paths. In fact, after performing
the two steps described above, each vertex of $G''$ has either degree
$1$ or degree $2$. Also, every vertex of degree $2$ is the middle
vertex of either a leaf-path or an edge-path.  Hence, no cycle is
created. Further, $C''(G'',\mathcal{T}'')$ is equivalent to $C'$, as
in any c-planar drawing of $C''$ a vertex $t$ that has been removed
from a cluster $\tau$ can be reinserted inside $R(\tau)$ and connected
to all the vertices of $\tau$ while maintaining c-planarity (the other
direction being trivial).

We now show how to construct instance $C$ starting from $C''$. For
each vertex $v^i_\alpha$ whose parent $\tau^i_\alpha$ in
$\mathcal{T}''$ is different from $\mu_\alpha$, we subdivide edge
$(v^i_\alpha,v^i)$ with a vertex $z^i_\alpha$; we add $z^i_\alpha$ to
cluster $\mu_\alpha$; and we remove $\tau^i_\alpha$ from the children
of $\mu_\alpha$ and add $\tau^i_\alpha$ as a child of the root
$\lambda$. For each vertex $v^i_\beta$ whose parent $\tau^i_\beta$ in
$\mathcal{T}''$ is different from $\mu_\beta$, we subdivide edge
$(v^i_\beta,v^i)$ with a vertex $z^i_\beta$; we add $z^i_\beta$ to
cluster $\mu_\beta$; and we remove $\tau^i_\beta$ from the children of
$\mu_\beta$ and add $\tau^i_\beta$ as a child of the root $\lambda$.

Let $\mu'$ and $\mu''$ be the parent clusters of the $3$ copies of
$u'_\rho$ and $u''_\rho$, respectively, in $\mathcal{T}''$. Subdivide
the edge connecting a vertex in $\mu'$ to $u'_\beta$ with a new vertex
and the edge connecting a vertex in $\mu''$ to $u''_\beta$ with a new
vertex, and add both such vertices to $\mu_\rho$. Also, remove $\mu'$
and $\mu''$ from the children of $\mu_\rho$ and add them as children
of the root $\lambda$.

Further, as long as there exists a cluster $\mu\neq \mu_\rho \in
\mathcal{T}''(\mu_\rho)$ such that all the children of $\mu$ are
leaves, we perform the following procedure. We add a new cluster
$\mu'$ to $\mathcal{T}''$ as a child of the root $\lambda$. Consider
the parent $\nu$ of $\mu$ in $\mathcal{T}''$. For each vertex $v^i \in
\mu$, we remove $v^i$ from the children of $\mu$ and add it as a child
of $\nu$; also, we subdivide the unique edge $(v^i_\beta,x)$ incident
to $v^i_\beta$ with a new vertex that we add to cluster
$\mu'$. Finally, we remove $\mu$ from $\mathcal{T}''$. The instance
$C(G,T)$ obtained by applying the reduction to \csefep instance of
Fig.~\ref{fig:sefeinstance} can be seen in Fig.~\ref{fig:B-flat}. In
order to prove that $C$ is equivalent to $C''$, observe that paths
$(v^i_\alpha,z^i_\alpha,v^i,z^i_\beta,\dots,v^i_\beta)$ obtained from
leaf-paths $(v^i_\alpha,v^i,v^i_\beta)$ are bounded to cross the
boundary of $R(\mu_\alpha)$ in $C$ in the same order in which the
corresponding leaf-paths are bounded to cross the boundary of
$R(\mu_\alpha)$ in $C''$. Namely, for each cluster $\mu \in
\mathcal{T}''$ there exists a cluster $\mu' \in \mathcal{T}$ imposing
the same consecutivity constraint on the ordering in which paths cross
the boundary of $R(\mu_\alpha)$.  This concludes the proof of the
theorem. \qed }

\ifllncsvar
\begin{proof}
  \proofthflat
\end{proof}
\fi

\begin{corollary}\label{co:flat}
  \ptcTWOpbepshort $\propto$ \cpp with flat cluster hierarchy and
  underlying graph that is a set of paths.
\end{corollary}

\section{Conclusions and Open Problems}\label{se:conclusions}
In this paper we show that \csefep is polynomial-time reducible to
\cpp even in the case in which the cluster hierarchy is flat and the
underlying graph is a set of paths.

We regard as an intriguing open question whether a polynomial-time
reduction exists from general instances of \sefep to instances of
\cpp, which would prove, together with the reduction by
Schaefer~\cite{s-ttphtpv-13}, these two problems to be ultimately the
same. Moreover, as our reduction produces instances of \cpp with a
number of clusters depending linearly in the size of the reduced
\csefep instance, it is worth of interest asking whether a sublinear
or constant number of clusters would suffice.

{\bibliography{bibliography}}

\ifllncsvar 
\else

\newpage

\section*{\large Appendix}

In this appendix we report additional details not included in the
extended abstract due to space reasons.

\subsection*{\bf Additional material for the proof of
  Theorem~\ref{th:non-flat}}

\begin{figure}[htb]
  \centering
  \includegraphics[width=.45\textwidth]{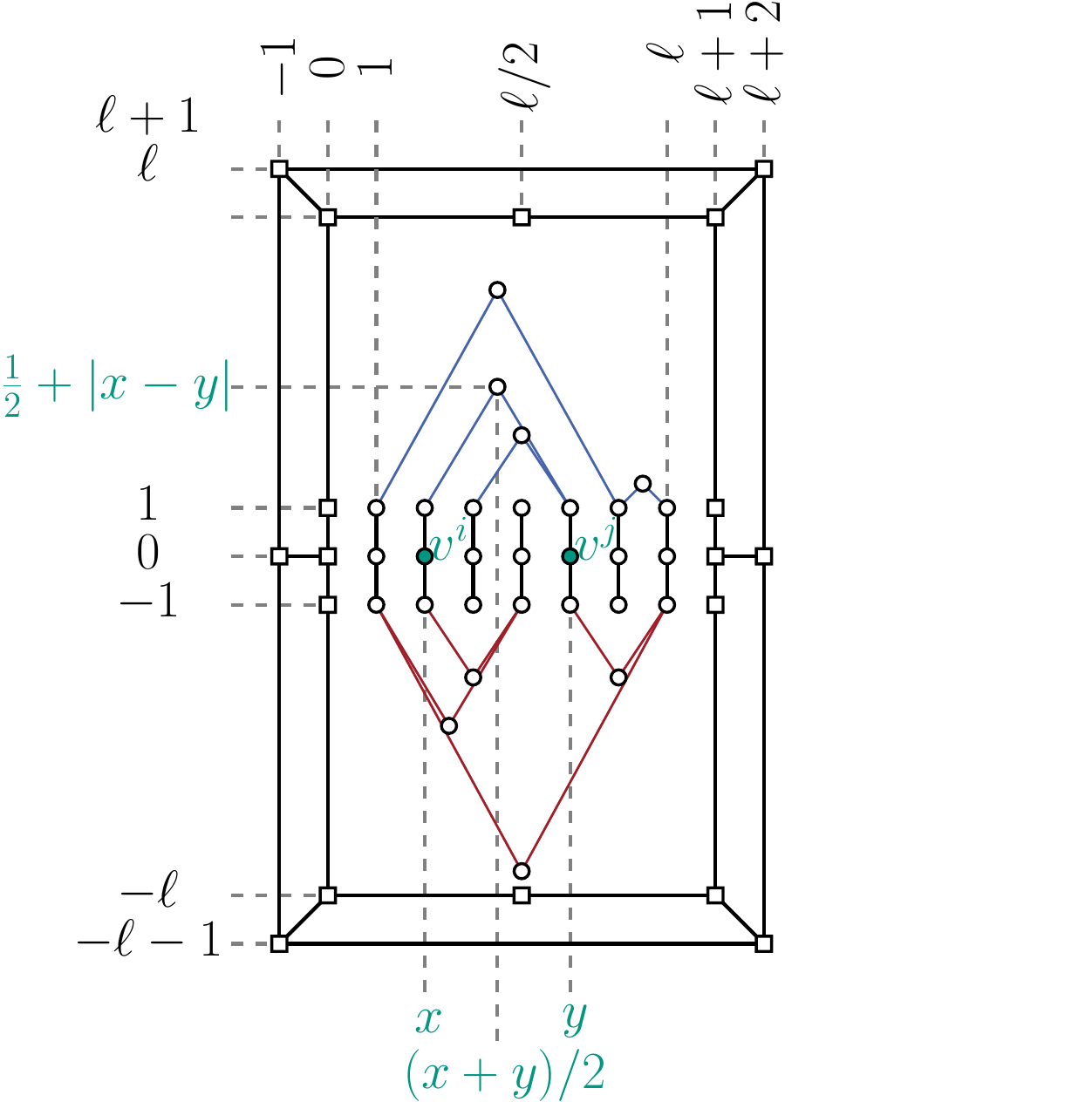}
  \caption{Construction of a c-planar drawing of $C(G,\mathcal{T})$
    starting from \sefesolution, where $x=\phi(v_i)$ and
    $y=\phi(v_j)$.}\label{fig:sefe-cp-drawing}
\end{figure}

\subsection*{\bf Full proof of Theorem~\ref{th:flat}}

\medskip
\noindent{\bf Theorem~\ref{th:flat}.} \emph{ \csefep $\propto$ \cpp
  with flat cluster hierarchy and underlying graph that is a set of
  paths.  }

\medskip
\begin{proof}
  \proofthflat
\end{proof}

\fi
\end{document}